\documentclass[runningheads]{llncs}

\usepackage[margin=0.9in]{geometry}

\usepackage{algorithmicx}
\usepackage[noend]{algpseudocode}
\usepackage{algorithm}

\algnewcommand{\algorithmicswitch}{\textbf{switch}}
\algdef{SE}[SWITCH]{Switch}{EndSwitch}[1]{\algorithmicswitch\ #1\ \algorithmicdo}{\algorithmicend\ \algorithmicswitch}%
\algtext*{EndSwitch}%

\algnewcommand{\algorithmiccase}{\textbf{case}}
\algdef{SE}[CASE]{Case}{EndCase}[1]{\algorithmiccase\ #1}{\algorithmicend\ \algorithmiccase}%
\algtext*{EndCase}%

\algnewcommand{\algorithmicon}{\textbf{on}}
\algdef{SE}[ON]{On}{EndOn}[1]{\algorithmicon\ #1\ \algorithmicdo}{\algorithmicend\ \algorithmicon}%
\algtext*{EndOn}%

\algrenewcommand{\algorithmicdo}{}
\algrenewcommand{\algorithmicthen}{}

\algnewcommand{\algorithmicgoto}{\textbf{goto}}%
\algnewcommand{\Goto}[1]{\algorithmicgoto~\ref{#1}}%

\algnewcommand{\algorithmicbreak}{\textbf{break}}%
\algnewcommand{\Break}[0]{\algorithmicbreak}%

\algnewcommand{\algorithmicwaiton}{\textbf{wait on}}%
\algnewcommand{\WaitOn}[1]{\algorithmicwaiton~{#1}}%

\PassOptionsToPackage{hyphens}{url}
\usepackage{url}
\usepackage{hyperref}
\hypersetup{breaklinks=true}

\usepackage[utf8]{inputenc}
\usepackage[T1]{fontenc}

\usepackage{amssymb,amsfonts,amsmath,amsthm}
\usepackage{graphicx}
\usepackage{xcolor}

\usepackage{subcaption}
\usepackage{float}

\usepackage{IEEEtrantools}
\allowdisplaybreaks

\usepackage[shortlabels,inline]{enumitem}

\usepackage{xstring}

\usepackage{datetime2}

\usepackage{siunitx}

\usepackage{varwidth}

\usepackage{tikz}
\usetikzlibrary{calc}
\usetikzlibrary{arrows}
\usetikzlibrary{arrows.meta}
\usetikzlibrary{patterns}
\usetikzlibrary{positioning}
\usetikzlibrary{decorations.pathreplacing}
\usetikzlibrary{shapes.misc}
\usetikzlibrary{spy}

\pgfdeclarelayer{bg1}
\pgfdeclarelayer{bg2}
\pgfsetlayers{bg1,bg2,main}

\usepackage{pgfplots}
\pgfplotsset{compat=1.14}
\usepgfplotslibrary{fillbetween}

\usepackage{xspace}
\usepackage{ifthen}

\newcommand{\LOGda}[2]{%
    \ifthenelse{\equal{#1}{}}{%
        \ensuremath{\mathsf{LOG}_{\mathrm{da}}^{#2}}%
    }{%
        \ensuremath{\mathsf{LOG}_{\mathrm{da},#1}^{#2}}%
    }%
}
\newcommand{\LOGbft}[2]{%
    \ifthenelse{\equal{#1}{}}{%
        \ensuremath{\mathsf{LOG}_{\mathrm{bft}}^{#2}}%
    }{%
        \ensuremath{\mathsf{LOG}_{\mathrm{bft},#1}^{#2}}%
    }%
}

\newcommand{\eg}[0]{\emph{e.g.}\xspace}

\theoremstyle{plain}
\newtheorem*{theorem*}{Theorem}

\definecolor{myParula01Blue}{RGB}{0,114,189}
\definecolor{myParula02Orange}{RGB}{217,83,25}
\definecolor{myParula03Yellow}{RGB}{237,177,32}
\definecolor{myParula04Purple}{RGB}{126,47,142}
\definecolor{myParula05Green}{RGB}{119,172,48}
\definecolor{myParula06LightBlue}{RGB}{77,190,238}
\definecolor{myParula07Red}{RGB}{162,20,47}

\tikzset{myparula11/.style={color=myParula01Blue,solid,mark=+,mark options={solid}}}
\tikzset{myparula12/.style={color=myParula01Blue,densely dashed,mark=x,mark options={solid}}}
\tikzset{myparula13/.style={color=myParula01Blue,densely dotted,mark=o,mark options={solid}}}
\tikzset{myparula14/.style={color=myParula01Blue,dashdotted,mark=triangle,mark options={solid}}}
\tikzset{myparula15/.style={color=myParula01Blue,dashdotdotted,mark=square,mark options={solid}}}

\tikzset{myparula21/.style={color=myParula02Orange,solid,mark=+,mark options={solid}}}
\tikzset{myparula22/.style={color=myParula02Orange,densely dashed,mark=x,mark options={solid}}}
\tikzset{myparula23/.style={color=myParula02Orange,densely dotted,mark=o,mark options={solid}}}
\tikzset{myparula24/.style={color=myParula02Orange,dashdotted,mark=triangle,mark options={solid}}}
\tikzset{myparula25/.style={color=myParula02Orange,dashdotdotted,mark=square,mark options={solid}}}

\tikzset{myparula31/.style={color=myParula03Yellow,solid,mark=+,mark options={solid}}}
\tikzset{myparula32/.style={color=myParula03Yellow,densely dashed,mark=x,mark options={solid}}}
\tikzset{myparula33/.style={color=myParula03Yellow,densely dotted,mark=o,mark options={solid}}}
\tikzset{myparula34/.style={color=myParula03Yellow,dashdotted,mark=triangle,mark options={solid}}}
\tikzset{myparula35/.style={color=myParula03Yellow,dashdotdotted,mark=square,mark options={solid}}}

\tikzset{myparula41/.style={color=myParula04Purple,solid,mark=+,mark options={solid}}}
\tikzset{myparula42/.style={color=myParula04Purple,densely dashed,mark=x,mark options={solid}}}
\tikzset{myparula43/.style={color=myParula04Purple,densely dotted,mark=o,mark options={solid}}}
\tikzset{myparula44/.style={color=myParula04Purple,dashdotted,mark=triangle,mark options={solid}}}
\tikzset{myparula45/.style={color=myParula04Purple,dashdotdotted,mark=square,mark options={solid}}}

\tikzset{myparula51/.style={color=myParula05Green,solid,mark=+,mark options={solid}}}
\tikzset{myparula52/.style={color=myParula05Green,densely dashed,mark=x,mark options={solid}}}
\tikzset{myparula53/.style={color=myParula05Green,densely dotted,mark=o,mark options={solid}}}
\tikzset{myparula54/.style={color=myParula05Green,dashdotted,mark=triangle,mark options={solid}}}
\tikzset{myparula55/.style={color=myParula05Green,dashdotdotted,mark=square,mark options={solid}}}

\tikzset{myparula61/.style={color=myParula06LightBlue,solid,mark=+,mark options={solid}}}
\tikzset{myparula62/.style={color=myParula06LightBlue,densely dashed,mark=x,mark options={solid}}}
\tikzset{myparula63/.style={color=myParula06LightBlue,densely dotted,mark=o,mark options={solid}}}
\tikzset{myparula64/.style={color=myParula06LightBlue,dashdotted,mark=triangle,mark options={solid}}}
\tikzset{myparula65/.style={color=myParula06LightBlue,dashdotdotted,mark=square,mark options={solid}}}

\tikzset{myparula71/.style={color=myParula07Red,solid,mark=+,mark options={solid}}}
\tikzset{myparula72/.style={color=myParula07Red,densely dashed,mark=x,mark options={solid}}}
\tikzset{myparula73/.style={color=myParula07Red,densely dotted,mark=o,mark options={solid}}}
\tikzset{myparula74/.style={color=myParula07Red,dashdotted,mark=triangle,mark options={solid}}}
\tikzset{myparula75/.style={color=myParula07Red,dashdotdotted,mark=square,mark options={solid}}}

\pgfplotsset{
    mysimpleplot/.style = {
        every axis plot/.prefix style={thick},
        width=1.0\linewidth,
        height=0.75\linewidth,
        title style={font=\footnotesize,align=center},
        legend cell align=left,
        legend style={font=\footnotesize},
        legend columns=3,
        legend style={
            at={(0.5,1)},
            yshift=0.3em,
            anchor=south,
            draw=none,
            /tikz/every even column/.append style={
                column sep=0.3em
            },
            cells={
                align=left
            }
        },
        grid=both,
        minor tick num=3,
        major grid style={solid,draw=gray!50},
        minor grid style={densely dotted,draw=gray!50},
        label style={font=\footnotesize,align=center},
        tick label style={font=\footnotesize},
    },
}

\usepackage{comment}

\usepackage{caption}

\usepackage{tikz}
\usetikzlibrary{shapes.geometric, arrows}

\tikzstyle{block} = [rectangle, rounded corners, minimum width=3cm, minimum height=1cm,text centered, draw=black, fill=white]
\tikzstyle{arrow} = [thick,<-,>=stealth]

\usepackage{xcolor}
\hypersetup{
    colorlinks,
    linkcolor={red!50!black},
    citecolor={blue!50!black},
    urlcolor={blue!80!black}
}
\usepackage{hyperref}

\setlist[itemize]{label=\textbullet}

\begin{document}
\title{MEV Capture and Decentralization in Execution Tickets}
%
%
\author{
Jonah Burian\inst{1} \and
Davide Crapis\inst{2} \and
Fahad Saleh \inst{3}
}

\institute{
Blockchain Capital\\
\email{jonah@blockchaincapital.com} \and
Robust Incentives Group - Ethereum Foundation\\
\email{davide.crapis@ethereum.org} \and 
University of Florida\\
\email{fahad.saleh@ufl.edu} 
}
\maketitle 

\begin{abstract}
We provide an economic model of Execution Tickets and use it to study the ability of the Ethereum protocol to \textit{capture MEV} from block construction. We demonstrate that Execution Tickets extract all MEV when all buyers are homogeneous, risk neutral and face no capital costs. We also show that MEV capture decreases with risk aversion and capital costs. Moreover, when buyers are heterogeneous, MEV capture can be especially low and a single dominant buyer can extract much of the MEV. This adverse effect can be partially mitigated by the presence of a Proposer Builder Separation (PBS) mechanism, which gives ET buyers access to a market of specialized builders, but in practice centralization vectors still persist. With PBS, ETs are concentrated among those with the highest ex-ante MEV extraction ability and lowest cost of capital. We show how it is possible that large investors that are not builders but have substantial advantage in capital cost can come to dominate the ET market.
\end{abstract}

\section{Introduction}
\label{sec:introduction}
Within the current Ethereum blockchain, a block proposer possesses a short-term monopoly right to propose the Execution Payload.\footnote{The block proposer refers to the validator chosen in the PoS random leader election to ``propose'' the next block. Moreover, the Execution Payload is the object in the block that contains the ordered list of transactions. For additional details regarding the Ethereum protocol, the interested reader may consult \cite{john2024economics}.} Notably, this monopoly right confers discretion over block contents and that discretion can be leveraged to acquire additional value, known as \textit{Maximal Extractable Value (MEV)}. Importantly, MEV is generally extracted from users and thereby serves as a disincentive for users to interact with the Ethereum blockchain. As a consequence, there is an active discussion regarding methods that enable the Ethereum protocol to capture MEV so that the proceeds can be distributed in such a way as to mitigate user losses.

Protocol MEV capture has been an open problem for some time. Recently Attester Proposer Separation (APS) was proposed, which is the idea of separating Execution Payload proposal rights from other duties of validators and selling the rights to a market of buyers. This mechanism has the potential, if implemented correctly, to preserve decentralization of Ethereum's validator set while enabling MEV capture. One of the most prominent implementation proposals is \textit{Execution Tickets (ETs)}. The primary contribution of this paper is to provide the first equilibrium analysis of ETs. This analysis allows us to shed light both on the extent to which ETs can successfully capture MEV and also implications regarding centralization.

MEV was first documented by \cite{daian2020flash} and corresponds specifically to value that can be extracted from users through discretion in determining which transactions are included in a block and their ordering.\footnote{The interested reader may consult \cite{angeris2023specter} for a formalization of MEV.} A classic example of MEV is a sandwich attack whereby a user order at a Decentralized Exchange (DEX) is sandwiched between a front-run and a back-run, yielding a profit to the trader conducting the sandwich but a loss for the user who faces higher trading costs due to the front-run (see \cite{harvey2024evolution} for details). In the context of a sandwich attack, a successful MEV internalization technique would extract the gain from the trader conducting the sandwich attack so that it could be redistributed to mitigate the effects of such attacks on users.

Under ET mechanism, there is a lottery in each slot to determine who will have the right to propose an Execution Block or Payload. This lottery is separate from the Beacon proposer lottery which identifies the consensus proposer who has the right to propose a Consensus Block, based on stake weight. The lottery tickets in the Execution Proposer lottery are ETs. Notably, an ET is a valid ticket to all future execution lotteries until it wins a lottery at which time it is removed from circulation. In turn, an ET always confers the right to propose the Execution Payload for exactly one slot in the future but the slot number is a random variable with a geometric distribution.

\begin{tikzpicture}[node distance=5cm, align=center]

\node[style=block] (validator) {Validator};
\node[style=block] (etbuyer) [right of=validator] {ET Buyer};
\node[style=block] (blockbuilder) [right of=etbuyer] {Block Builder};

\draw [arrow] (validator) -- node[above] {APS} (etbuyer);
\draw [arrow] (etbuyer) -- node[above] {PBS} (blockbuilder);

\node[below of=validator, yshift=3.6cm, align=center] {\textit{Attestations} \& \\ \textit{Consensus Block Proposal}};
\node[below of=etbuyer, yshift=3.6cm, align=center] {\textit{Execution Payload Proposal}};
\node[below of=blockbuilder, yshift=3.6cm, align=center] {\textit{Ordered List} \\ \textit{of Transactions}};

\end{tikzpicture}

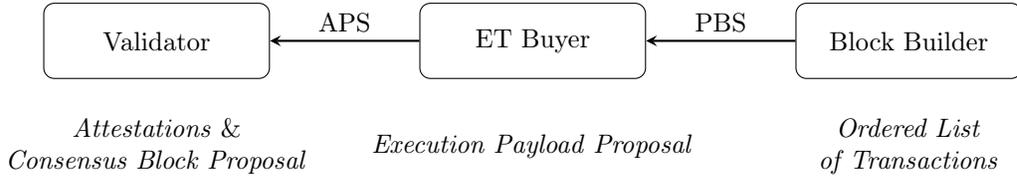
\captionof{figure}{This diagram illustrates the separation of duties between Validator, ET Buyer, and Block Builder roles in the presence of APS and PBS mechanisms (duties are below each entity in cursive).}
\hspace{0.7cm}

Our paper provides three important \textit{insights regarding protocol MEV capture via ETs}:
\begin{itemize}
    \item[(\textit{i})] Risk Aversion and Capital Costs undermine the success of ETs internalizing all MEV;
    \item[(\textit{ii})] Heterogeneity among builders with or without PBS undermines the success of ETs internalizing all MEV;
    \item[(\textit{iii})] In all situations, the more competitive the playing field, the more MEV can be internalized.
\end{itemize}

The first point arises because an ET confers the right to determine the block contents for a future slot at a random time. This structure introduces both risk in terms of the pay-off and costs in terms of needing to secure capital ahead of time. As a consequence, when ET buyers are risk-averse and face capital costs, they are not willing to pay the expected value of MEV for an ET. Rather, each buyer would pay the risk-adjusted discounted value of MEV which is necessarily lower than the expected value of MEV. We formalize this point through two results, Propositions \ref{prop:main1} and \ref{prop:main2}. Proposition \ref{prop:main1} demonstrates that the protocol does fully capture MEV when builders are risk-neutral and face no capital costs, whereas Proposition \ref{prop:main2} shows that the protocol does not fully capture MEV when builders are risk-averse and face non-zero capital costs.

The second point arises because, when buyers are heterogeneous, they possess different valuations for ETs. In turn, the buyer with the highest valuation has no incentive to pay more than the buyer with the second highest valuation and no other buyer has an incentive to pay above the valuation for the buyer with the second-highest valuation either. In turn, the gap between the two highest valuations is captured by the buyer with highest valuation and not the protocol. We demonstrate this result formally in Proposition \ref{prop:main3} and Proposition \ref{prop:main4}. 

In the presence of a (PBS) mechanism, \textit{any} ET holder can sell their ET in the MEV Boost market at a competitive valuation that does not vary by builder. That said, every builder also has the option to opt out of PBS (in some instances this can be done ex-post). Hence, buyers remain heterogeneous with PBS, and as before, MEV is similarly not fully internalized.

The third point arises given the intuition that the gap between the two highest valuations is captured by the buyer with highest valuation and not the protocol. The tighter this gap, the more MEV is captured.  \\

Our paper also provides implications regarding \textit{centralization vectors in the ET market}:
\begin{itemize}
    \item[(\textit{i})] When buyers are homogeneous, ET holdings are decentralized;
    \item[(\textit{ii})] When buyers are heterogeneous and there is no PBS mechanism, ET holdings are centralized;
    \item[(\textit{iii})] When buyers are heterogeneous and PBS is present, ET ownership is centralized among buyers who balance low capital costs with high MEV extraction abilities Moreover, large investors with lower capital costs may dominate the market. Even with ticket ownership concentration, the block-building rights are likely sold ex-post via PBS, and thus the tickets are likely exercised by builders with the best ex-post MEV extraction ability in a slot.
\end{itemize}

The first point arises because, when buyers are homogeneous, they possess identical valuations for ETs. In turn, all buyers are willing to purchase ETs. This result is stated formally as part of Propositions \ref{prop:main1} and \ref{prop:main2}. Crucially, as per the second point, \textit{any} heterogeneity among buyers whereby the ex ante valuation for ETs is higher for some set of buyers than others results in a centralization of ET holdings whereby only the buyers with the highest ex ante valuation hold ETs. This point is formalized by Proposition \ref{prop:main3} and Proposition \ref{prop:main4}. 

PBS gives buyers equal access to builders at the time of block construction, meaning ticket holders don't need to have native MEV extraction capabilities. Instead, they can solicit the block's value from the competitive PBS market. However, builders always have the option not to run PBS, which may be rational if their building ability surpasses that of the market. This dynamic maintains buyer heterogeneity while also providing those without native MEV extraction ability a baseline extraction capability via PBS. Thus, better capital costs and superior native MEV extraction ability can lead buyers to have a higher ex-ante valuation of ETs than others. In Proposition \ref{prop:main8}, we formalize the criteria for balancing capital costs with MEV extraction ability (which can be sourced through PBS), leading certain buyers to dominate the market. In Proposition \ref{prop:main9}, we further formalize that large investors with systematically better capital costs might come to dominate the ET buying market. 

Given PBS, the owner of the ET might not be the one who exercises the building rights and may instead sell the right via PBS. Hence, ticket centralization does not necessarily imply builder centralization.\\


We proceed hereafter by providing comprehensive background in Section \ref{sec:background}. We then state our formal economic model in Section \ref{sec:model}. We provide the general equilibrium solution in Section \ref{sec:solution} along with several important special cases. Our results are provided in Section \ref{sec:results}. We discuss practical implications of our findings and forward-looking thoughts in Section \ref{sec:discussion}.

\section{Background}
\label{sec:background}

\subsection*{Current MEV Market Structure}

Without external aids, proposers struggle to capture the bulk of accessible MEV. This is because effective capture necessitates sophistication given the need to quickly and efficiently navigate through combinatorially complex search spaces. Acknowledging that validators might not be ideally suited for these intricate challenges, PBS \cite{barnabe_pbs} has been widely embraced \cite{pbs-paper}, with the most common instantiation being MEV-Boost.\footnote{See the MEV-Boost GitHub repository here \cite{mev_boost}.} This approach enables all proposers, irrespective of their level of sophistication, to capture the majority of the MEV value in their block.

PBS distinguishes the block-building function from the proposing function. Proposers can opt into an auction where builders bid for the right to control the ordering of transactions in a block. Proposers then include in the Execution Payload the list of transaction with the highest associated bid. Since the builders select the payload, they receive the MEV minus the bid, while the validator receives the bid. Effectively, proposers are auctioning off their one slot monopoly to builders. Given the competitiveness of the market and the short-term monopoly the validator has on block production, the winning bid ends up being slightly less than the value of the MEV.

\subsection*{Problems with the Current Market Structure}

\subsubsection{Centralization Problem.}

The current MEV supply chain puts centralization pressure\footnote{
Maintaining a robust Ethereum network requires validator decentralization, namely geographical, client, and owner diversity of the validator set. Systematic advantages for certain validators and externalities that encourage colocation corrode the foundational integrity of the network. Builder centralization, on the other hand, is problematic if it leads to systematic censorship of transactions. A decentralized validator set may be able to maintain censorship resistance through mechanisms like inclusion lists \cite{eip-7547}, meaning that builder centralization may not be an issue as long as there is a robust decentralized validator set. 
} on Ethereum as more sophisticated validators consistently make more money than less sophisticated ones \cite{centralization-pb}. This is because there is an incentive for proposers to play timing games \cite{timinggames}, and certain validators are better equipped to play these games. Given differentials in average proposer rewards, the current market structure should, in theory, lead to the centralization of the validator set in the long run.\footnote{The intuition behind the existence of timing games in the MEV instance is that capturable MEV monotonically increases over time. Proposers are incentivized to wait until the last instant to propose so that they get the necessary attestations to get their block accepted. Sophisticated proposers should strategically collocate with attesters and builders so that they can accept a bid from a builder as late as possible and push the block through as quickly as possible. Not all validators have the infrastructure and/or relationships to play these games.} As it stands, builders are centralized\footnote{A handful of builders dominate most of the blocks, \eg, beaverbuild, Titan Build, rsync \cite{mev_pics}} and validators centralizing with them can undermine the overall decentralization of the network.

\subsubsection{Allocation Problem.}

In the current market structure, proposers earn nearly all the MEV, yet they are merely agents of the protocol. MEV is generated by users on applications that sit on the protocol. From an allocation perspective, MEV value should probably be allocated to users, applications, or the protocol. The proposer is an agent of the protocol; it does not make sense that all the value is directed to them.

Applications are becoming increasingly aware that they are leaving value on the table. New and existing applications are designing/redesigning their protocols to leak less MEV, with some planning to redistribute this value to users. That said, it is unlikely that MEV will cease to exist. Applications might not be able to internalize all their MEV, there will most likely continue to be naive applications that still produce MEV, and there will be cross-domain (cross-application) MEV. Hence, Ethereum will likely always have residual MEV. Under the current market structure, this value will flow to validators. Arguably, this value should be captured by the protocol and redistributed according to the protocol's macroeconomic goals, \eg, given pro rata to stakers or Ethereum holders.

\subsection*{Potential Solutions}
Two approaches have recently emerged as potential avenues for solving the centralization and allocation problems. These are mechanisms where the right to propose an Execution Payload is not given for free to the Beacon proposer, but it is instead sold separately to an execution proposer. The two mechanisms are Execution Auctions and Execution Tickets \cite{futuremev}.\footnote{Burian and Crapis have written about these two solutions on Eth Research. See \cite{burian_crapis_tickets} and \cite{burian_crapis_auctions}.}

\subsubsection{Execution Auctions (EAs)}

The right to propose an Execution Payload is deterministically allocated in advance for each slot, the slot execution proposer can purchase this right by bidding in a slot auction held beforehand (\eg, 32). EAs are essentially slot auctions carried out in advance.

\subsubsection{Execution Tickets (ETs)}

With ETs, the execution proposing right is non-deterministically allocated, unlike with EAs. Proposers can purchase a lottery ticket (ET) from the protocol. In the steady state, there are \(n\) tickets. Before each slot, a winner is drawn at random from the ticket pool and receives the right to propose. The winning ticket is removed from the pool (burned), and a new one is sold, maintaining the invariant of there always being \(n\) tickets in the lottery pool. The simple version of the protocol gives the winner the right to propose the following block.\footnote{Colloquially, ETs with a single slot resolution are called sETs (simple ETs).} The original Execution Ticket post suggested a general version of the protocol where the winner has the right to propose \(m\) slots later (\eg, 32) \cite{drake2023execution}.

\section{Model}
\label{sec:model}
We model a finite set of buyers, $\mathcal{B} = \{1,...,B\}$, who optimally select the number of ETs to hold. In more detail, for each slot $t \in \mathbb{N}$ of the Ethereum protocol, each buyer $b \in \mathcal{B}$ optimally selects a number of Execution Tickets, $k_b \in \{0,...,N\}$, at the end of the previous slot $t-1$ where $N \in \mathbb{N}_+$ is the total number of tickets in circulation. More formally, buyer $b$ maximizes their risk-adjusted expected profit less capital costs as follows:

\begin{equation}
        \underset{k_b \in \{0,...,N\} }{\max}~ \mathbb{E}[\Pi_b(\text{P\&L}_{b,t})] - r_{b} \cdot P \cdot k_b
    \label{eqn:ETpurchase}
\end{equation}
where $\text{P\&L}_{b,t}$ denotes the net profit of buyer $b$ in slot $t$, $\Pi_b(x)$ is a function for buyer $b$'s risk-adjustment, $r_b \geq 0$ denotes buyer $b$'s cost of capital, and $P \geq 0$ denotes the endogenous stationary price of Execution Tickets. For simplicity, we assume that MEV extraction for each buyer is i.i.d. across time and thus the endogenous stationary price of ETs, $P$, is also the time-invariant price for each time $t$. We discuss relaxing this assumption in Section \ref{sec:discussion}. Additionally, as per prior literature, we assume that $\Pi(x)$ is a twice-continuously differentiable function that satisfies $\Pi^\prime(x) > 0$ and $\Pi^{\prime \prime}(x) \leq 0$ over $x \in \mathbb{R}_+$ and $\Pi(x) = 0$ whenever $x \leq 0$. $\text{P\&L}_{b,t}$ is given explicitly as follows:

\begin{equation}
    \text{P\&L}_{b,t} = \mathcal{I}_{b,t} \cdot (R_{b,t} - P)
    \label{eqn:payoff}
\end{equation}
where $\mathcal{I}_{b,t}$ denotes the probability that an Execution Ticket of buyer $b$ is selected and $R_{b,t}$ is a non-negative random variable denoting the pay-off accrued during slot $t$ by buyer $b$ (if buyer $b$ is selected for that slot) where we impose $\mathbb{E}[R_{b,t}] < \infty$.\footnote{In our initial results, we abstract from MEV-Boost and, in this case, $R_{b,t}$ represents the value extracted directly by buyer $b$ in slot $t$. For Proposition \ref{prop:main8}, we allow for an MEV Boost market in which case the buyer receives the sales proceeds of the ET and thus $R_{b,t}$ corresponds to the sales proceeds from MEV Boost.} In more detail, the portfolio value of buyer $B$ at the end of slot $t-1$, denoted $V_{t-1}$, after purchasing $k_b$ Execution Tickets is given as follows:

\begin{equation}
    V_{t - 1} = k_b \cdot P
    \label{eqn:vt-1}
\end{equation}
whereas, after the lottery in slot $t$, the portfolio value becomes:

\begin{equation}
    V_t = \mathcal{I}_{b,t} \cdot ((k_b - 1) \cdot P + R_{b,t}) + (1 - \mathcal{I}_{b,t}) \cdot (k_b \cdot P),\qquad k_b \geq 1
    \label{eqn:vt}
\end{equation}

Equations \eqref{eqn:vt-1} and \eqref{eqn:vt} directly imply Equation \eqref{eqn:payoff}. Intuitively, a buyer's Execution Ticket could be sold at market value $P$ but being selected entails the ticket being burned in return for a random pay-off determined by the MEV extracted by the buyer from the block, $R_{b,t}$. Thus, being selected entails a random P\&L increasing in the block MEV and decreasing in the price of the Execution Ticket (formalized by Equation \ref{eqn:payoff}).

\label{sec:solution}
Before providing the general model solution, it is useful to put forth some preliminary quantities. In particular, we define the maximal price at which buyer $b \in \mathcal{B}$ would hold an Execution Ticket, $\overline{P}_b$, is given as follows:\footnote{Equation \eqref{eqn:maxPb} is always well-defined because $\{ P \geq 0 : \frac{1}{N} \mathbb{E}[\Pi_b(R_{b,t}- P)] - r_b \cdot P \geq 0\}$ is non-empty, bounded from above and closed. It is non-empty because $P = 0$ is always within the set, and it is bounded from above because all set elements satisfy $P \leq \frac{\mathbb{E}[\Pi_b(R_{b,t})]}{r_b \cdot N}$. Additionally, to see that it is closed, consider a convergent sequence, $\{ P_n \}_{n = 1}^{\infty} \subseteq \{ P \geq 0 : \frac{1}{N} \mathbb{E}[\Pi_b(R_{b,t}- P)] - r_b \cdot P \geq 0\}$ with limit point $P_\star = \lim_{n \to \infty} P_n$. Then, $\frac{1}{N} \mathbb{E}[\Pi_b(R_{b,t}- P_n)] - r_b \cdot P_n \geq 0$ for all $n$ so that we can take $n \to \infty$ on both sides which yields $\frac{1}{N} \lim_{n \to \infty} \mathbb{E}[\Pi_b(R_{b,t}- P_n)] - r_b \cdot P_\star = \frac{1}{N} \mathbb{E}[\Pi_b(R_{b,t}- P_\star)] - r_b \cdot P_\star  \geq 0$, implying $P_\star \in \{ P \geq 0 : \frac{1}{N} \mathbb{E}[\Pi_b(R_{b,t}- P)] - r_b \cdot P \geq 0\}$ and thereby establishing that $\{ P \geq 0 : \frac{1}{N} \mathbb{E}[\Pi_b(R_{b,t}- P)] - r_b \cdot P \geq 0\}$ is closed. As a technical aside, the interchange of the limit and expectation in the last equality is always valid due to the Dominated Convergence Theorem. In more detail, $\Pi_b(R_{b,t})$ is a uniform and integrable bound for $\{ \Pi_b(R_{b,t} - P_n) \}_{n \in \mathbb{N}}$; more explicitly, $\sup_{n \in \mathbb{N}} |\Pi_b(R_{b,t} - P_n)| \leq \Pi_b(R_{b,t})$ because $\Pi(x)$ is increasing and non-negative, and $\mathbb{E}[|\Pi_b(R_{b,t})|] =\mathbb{E}[\Pi_b(R_{b,t})] \leq \Pi_b(\mathbb{E}[R_{b,t}]) < \infty$ where the first inequality is Jensen's inequality relying on concavity of $\Pi(x)$ and the last inequality is due to $R_{b,t}$ having a finite first moment.}
\begin{equation}
\overline{P}_b = \max\{ P \geq 0 : \frac{1}{N} \mathbb{E}[\Pi_b(R_{b,t}- P)] - r_b \cdot P \geq 0\}
\label{eqn:maxPb}
\end{equation}
In turn, we specify the maximum price at which any buyer would purchase an Execution Ticket as follows:

\begin{equation}
    \overline{P}_{(1)} = \underset{b: b \in \mathcal{B}}{\max}~\overline{P}_b
\end{equation}
and we specify the set of buyers that would hold an Execution Ticket at this maximal price as follows:
\begin{equation}
\overline{\mathcal{B}} = \{ b \in \mathcal{B} : \overline{P}_b = \overline{P}_{(1)} \}
\label{eqref:bestbuyerset}
\end{equation}
Finally, we define the second largest value from $\{ \overline{P}_b \}_{b \in \mathcal{B}}$ as follows:
\begin{equation}
    \overline{P}_{(2)} = \begin{cases}
        \overline{P}_{(1)} & \text{if } |\overline{\mathcal{B}}| > 1\\
        \underset{b : b \in \mathcal{B} \bigcap \overline{\mathcal{B}}^c}{\max}~\overline{P}_b & \text{otherwise}
    \end{cases}
\end{equation}

With the given preliminaries, we can now state our general equilibrium solution:

\begin{proposition} Equilibrium Solution\\
In general, there exist multiple equilibria. More formally, necessary and sufficient conditions for an equilibrium are that the stationary price, $P$, and ET holdings, $\{ k_b \}_{b \in \mathcal{B}}$ satisfy the following conditions:
\begin{itemize}
    \item[(i)] \underline{Execution Ticket Price}\\
    The ET price, $P$, satisfies $P \in [\overline{P}_{(2)}, \overline{P}_{(1)}]$
    \item[(ii)] \underline{Execution Ticket Holdings}\\
    - Any buyer with a maximal price below the largest maximal price holds no ETs: $b \notin \overline{\mathcal{B}} \implies k_b = 0$\\
    - The set of buyers with maximal price equal to the largest 
    maximal price holds all ETs: $\sum\limits_{b : b \in \overline{\mathcal{B}}} k_b = N$
\end{itemize}
\label{prop:eqsoln}
\end{proposition}

An immediate corollary is that, if all buyers are homogeneous (i.e., $\forall b, b^\prime \in \mathcal{B}: R_{b,t} = R_{b^\prime, t}~a.s., r_b = r_{b^\prime}, \Pi_b = \Pi_{b^\prime}$), then there exists a unique equilibrium price for ETs, the maximal price that each buyer is willing to pay for an ET:

\begin{corollary} Homogeneous Equilibrium Solution\\
If all buyers are homogeneous $($i.e., $\forall b, b^\prime \in \mathcal{B}: R_{b,t} = R_{b^\prime, t}~a.s., r_b = r_{b^\prime}, \Pi_b = \Pi_{b^\prime})$, then there is a unique stationary price, $P$, and any feasible set of holdings that clears the market is an equilibrium. More explicitly, the equilibrium is given as follows:
\begin{itemize}
    \item[(i)] \underline{Execution Ticket Price}\\
    The ET price, $P$, is given as follows: $P = \overline{P}$ where $\overline{P} = \overline{P}_{(b)}$ for all $b \in \mathcal{B}$ due to homogeneity.
    \item[(ii)] \underline{Execution Ticket Holdings}\\
    ET holdings must satisfy only the market-clearing condition: $\sum\limits_{b : b \in \mathcal{B}} k_b = N$
\end{itemize}
\end{corollary}
A further corollary is that if buyers are homogeneous and risk-neutral (i.e., $\forall b: \Pi_b(x) = x$ and $\forall b, b^\prime \in \mathcal{B}: R_{b,t} \overset{d}{=} R_{b^\prime, t}, r_b = r_{b^\prime}, \Pi_b(x) = x$), then the unique ET price is available is simply the cost-of-capital-discounted value of the expected value of MEV from a single block:

\begin{corollary} Homogeneous Risk-Neutral Equilibrium Solution\\
If all buyers are homogeneous and risk-neutral $($i.e., $\forall b, b^\prime \in \mathcal{B}: R_{b,t} = R_{b^\prime, t} ~a.s., r_b = r_{b^\prime}, \Pi_b(x) = x)$, then there is a unique stationary price, $P$, and any feasible set of holdings that clears the market is an equilibrium. More explicitly, the equilibrium is given as follows:
\begin{itemize}
    \item[(i)] \underline{Execution Ticket Price}\\
    The ET price, $P$, is given as follows: $P = \overline{P}$ where $\overline{P} = \frac{\mathbb{E}[R_{b,t}]}{1 + r_b \cdot N}$ for all $b \in \mathcal{B}$ due to homogeneity.
    \item[(ii)] \underline{Execution Ticket Holdings}\\
    ET holdings must satisfy only the market-clearing condition: $\sum\limits_{b : b \in \mathcal{B}} k_b = N$
\end{itemize}
\end{corollary}

\section{Results}
\label{sec:results}
In this section we present our main results. We study two setups, first we focus on a pure ET mechanism and then we look at ET performance in the presence of PBS. In the first case, without PBS we find that: 

\begin{itemize}
    \item If ET buyers are homogeneous, risk-neutral and face no capital costs, then \textit{the protocol extracts all MEV and ET holdings are decentralized across buyers}.
    \item If investors are homogeneous but not necessarily risk-neutral and face capital costs, then \textit{the protocol does not extract all MEV and ET holdings are decentralized across buyers}.
    \item If buyers are heterogeneous, then \textit{MEV capture can be low, moreover, one buyer may extract most of the MEV }and the ET holding would be concentrated on that buyer. 
\end{itemize}

In the second case, with PBS, we find that: 

\begin{itemize}
    \item PBS turns the market into a special case of a no-PBS world with heterogeneous MEV extraction abilities and differing costs of capital. The protocol does not extract all MEV, leading to ET holdings being centralized among buyers who balance the best MEV extraction abilities with the lowest cost of capital.
    \item The ET holders may not necessarily exercise the building rights for the blocks and may instead outsource block construction via PBS.
\end{itemize}

To provide these results formally, we first clarify the meaning of MEV capture. In particular, the long-run proportion of MEV extracted by the protocol, $\chi$, is given as follows:

\begin{equation}
    \chi = \frac{P}{\mathbb{E}[R_{\overline{B}}]}
    \label{eqn:mevinternalizationratio}
\end{equation}
where $R_{\overline{B}}$ denotes the MEV extracted by the winning buyer in an arbitrary slot. Note that we can define $R_{\overline{B}}$ without explicitly specifying a slot because the distribution of MEV extraction by a winning buyer across time is i.i.d. This latter point follows from $\{ k_b \}_{b \in \mathcal{B}}$ being time-invariant and MEV extraction across buyers being i.i.d. We discuss relaxing these assumptions in Section \ref{sec:discussion}.

We now turn to stating each of our main results, in each of the subsequent sections. Our Section \ref{sec:nopbs-results} results are intended as conceptual benchmarks; we assume homogeneity/heterogeneity in MEV extraction ability in a vacuum, without specifying how such ability is formed, focusing instead on its impact on MEV capture and decentralization. Our Section \ref{sec:pbs-results} results are the most realistic, and we discuss further possible refinements in Section \ref{sec:discussion}.

\subsection{MEV Capture and Decentralization in  the ET Market}
\label{sec:nopbs-results}

\subsubsection*{Total MEV Capture w/ Decentralization.}

Our first main result establishes a benchmark setting whereby MEV capture is total and decentralization also arises:

\begin{proposition} Total MEV Capture with Decentralization\\
Assume that buyers are homogeneous, risk-neutral and face no capital costs $($i.e., $\forall b: r_b = 0, \Pi_b(x) = x$ and $\forall b, b^\prime: R_{b,t} = R_{b^\prime, t})$. Then, MEV capture by the protocol is total:

$$
\chi = 1
$$
Additionally, in this case, there exists an equilibrium with full decentralization of ET holdings:
$$
\forall b \in \mathcal{B}: k_b = \frac{N}{B}
$$
\label{prop:main1}
\end{proposition}

Proposition \ref{prop:main1} highlights that, when buyers possess identical MEV extraction skills and face no financing costs nor risk aversion, then the protocol extracts all MEV. Intuitively, when buyers possess identical MEV extraction skills, then they compete for ETs which results in ET pricing fully reflecting the risk-adjusted discounted value of the MEV from a future block. Moreover, when buyers are risk neutral, then there is no risk-adjustment and when buyers face no capital cost, then there is no discounting. Thus, when buyers possess identical MEV extraction skills while facing no risk aversion nor any capital costs, then all buyers are independently willing to pay the expected value of MEV from a future block and the competition to purchase ETs among the buyers ensures that the ET price exactly equals the expected value of MEV from a block. In turn, the protocol extracts all MEV.

\subsubsection*{Partial MEV Capture w/ Decentralization.}

Our second main result relaxes the assumptions of risk-neutrality and zero capital costs from Proposition \ref{prop:main1} but maintains homogeneity of buyers. In that context, we demonstrate that MEV capture is not total. Nonetheless, as in the previous case, decentralization arises in equilibrium:

\begin{proposition} Partial MEV Capture with Decentralization\\
Assume that buyers are homogeneous $(\forall b, b^\prime: R_{b,t} = R_{b^\prime, t}, r_b = r_{b^\prime}, \Pi_b = \Pi_{b^\prime})$ but risk-averse (i.e., $\forall b: \Pi_b^{\prime \prime} < 0$) and face non-zero capital costs $($i.e., $\forall b: r_b > 0)$ . Then, MEV capture by the protocol is partial:

$$
\chi < 1
$$
Nonetheless, in this case, there exists an equilibrium with full decentralization of ET holdings:
$$
\forall b \in \mathcal{B}: k_b = \frac{N}{B}
$$
\label{prop:main2}
\end{proposition}

Proposition \ref{prop:main2} arises because risk-aversion and non-zero capital costs imply that buyers do not value ETs as equivalent to the MEV value of a single block. In particular, since an ET promises MEV from a \textit{future} block, non-zero capital costs, $r > 0$, imply that the future MEV value is discounted at that non-zero rate $r > 0$. Moreover, since the MEV per block is uncertain, the ET value is further reduced by a risk-adjustment. Thus, although all buyers possess the same value for an ET under this setting, this value is lower than the expected value of MEV per block. Since the buyers are homogeneous, the market is competitive and the equilibrium ET price equals to the common risk-adjusted discounted MEV block value, implying that the protocol does not capture all MEV.

\subsubsection*{Low MEV Capture w/ Centralization.}

Our next main results highlight that when buyers are differentiated in their capital costs and/or MEV extraction abilities, then centralization arises. Moreover, the best buyer can acquire a large share of MEV at the expense of the protocol.

We look at two cases where buyers are risk-neutral, and highlight how heterogeneity with respect to either capital cost or MEV extraction ability results in lower MEV capture and higher centralization.


\begin{proposition}{Low MEV Capture w/ Heterogeneous Capital Costs}\\
Assume that buyers have homogeneous MEV extraction abilities $($i.e. $\forall b, b^\prime: R_{b,t} = R_{b^\prime, t})$, are risk-neutral and face heterogeneous capital costs.

In this case, there exists an equilibrium where MEV capture by the protocol is given as follows:

$$
\chi = \frac{1}{1 + r_{(2)} \cdot N} < 1
$$

where $r_{(2)}$ is the second-lowest cost-of-capital among all buyers. Moreover, the equilibrium ET distribution is determined fully by cost-of-capital:

$$
k_b = \begin{cases}
    \frac{N}{|\overline{\overline{\mathcal{B}}}|}
& \text{ if } r_b = r_{(1)}\\
0 & \text{if } r_b > r_{(1)}\end{cases} 
$$
where $r_{(1)}$ is the lowest cost of capital among all buyers.
\label{prop:main3}
\end{proposition}

Proposition \ref{prop:main3} highlights the fundamentally important role of capital costs. In more detail, the value of an ET to a buyer depends entirely on capital costs with buyers possessing lower capital costs possessing higher valuations. Then, as per the last part of Proposition \ref{prop:main3}, the ET market centralizes among the buyers with the lowest capital costs.

\begin{proposition}{Low MEV Capture w/ Heterogeneous MEV Extraction Ability}\\
Without loss of generality, assume that buyers are indexed by expected MEV extraction ability (i.e., $\forall b < b^\prime: \mathbb{E}[R_{b,t}] \geq \mathbb{E}[R_{b^\prime,t}]$. Assume further that the best MEV-extracting buyer is strictly better than the second-best MEV-extracting buyer (i.e., $\mathbb{E}[R_{1,t}] > \mathbb{E}[R_{2,t}]$). Moreover, for simplicity, assume that buyers are risk-neutral (i.e., $\forall b: \Pi_b(x) = x$) and face no capital costs (i.e., $\forall b: r_b = 0$). In this case, there exists an equilibrium where MEV capture by the protocol is given as follows:

$$
\chi = \frac{\mathbb{E}[R_{2,t}]}{\mathbb{E}[R_{1,t}]} < 1
$$
Additionally, the equilibrium ET distribution is fully centralized on the best buyer:

$$
k_1 = N,\qquad \forall b \neq 1: k_b = 0
$$
\label{prop:main4}
\end{proposition}

Proposition \ref{prop:main4} arises because, when buyers are heterogeneous in their MEV extraction abilities, then there exists an equilibrium where the ET price is equal to the expected MEV from the second-best buyer. In turn, the best buyer is able to buy the ET below their higher valuation and captures the difference. As a consequence, the protocol does not capture all the MEV. Crucially, to the extent that the best buyer is better than the second-best buyer, the best buyer is able to capture a commensurately larger share of MEV.

Another implication of Proposition \ref{prop:main4} is that centralization of ET holdings can arise when buyers are differentiated. More specifically, the best buyer is willing to pay more for an ET than all other buyers because the best buyer can generate higher MEV, on average, than any other buyer. In turn, the best buyer purchases all ETs. Of particular note, the relevant notion of being the best buyer is ex ante and not ex post. That is, even if a second buyer would be better at extracting MEV in a particular slot ex post, the sample path of events from the slot is not known when the ET is issued (or even when the ET is assigned to a slot). As a consequence, purchase decisions in the primary market are necessarily driven by ex ante valuations and thus so long as one buyer has a higher expected MEV value for a block, that buyer would dominate the market for ETs. As an aside, we note that our results do not preclude the best buyer selling the ET after the fact, but the key point is that such a sale (in the secondary market) would not raise revenue for the protocol but rather for the buyer who made the purchase in the primary market.

\begin{comment}
\begin{proposition}{Low MEV Capture w/ Centralization and Relaxed Assumptions}\\
Here we have the same set up as Proposition \ref{prop:main3} with the difference that buyers face heterogeneous capital costs. In this case, there exists an equilibrium where MEV capture by the protocol is given as follows:

$$
\chi = ? < 1
$$

Additionally, the equilibrium ET distribution is fully centralized on the set of buyers with:

$$
\min \{ \frac{E[R_{b,t}]}{1+r_bN} \}
$$
\label{prop:main5}
\end{proposition}

[TODO: ake some comment about the intuition of Proposition \ref{prop:main5}]
\end{comment}

\subsection{ETs in the Presence of PBS}
\label{sec:pbs-results}
In this section, we generalize our model to study outcomes in the presence of a PBS mechanism. Crucially, PBS allows any ET buyer to sell their ET in a market after MEV extraction abilities are realized. The ET buyer may sell in the PBS market or hold onto the ET, build the block and extract MEV as per their ability. More formally, in the presence of  PBS mechanism, the pay-off for an ET buyer $b \in \mathcal{B}$ is given as follows:
\begin{equation}
R_{b,t} = \max\{ \Gamma\big(\{X_{i,t} \}_{i \in \mathcal{B}}, \{Y_{j,t} \}_{j \in \beta}\big), X_{b,t}\}
\label{eqn:investormevpayoff}
\end{equation}
where $\{X_{i,t} \}_{i \in \mathcal{B}}$ corresponds to MEV extraction abilities of buyers, $\{Y_{j,t} \}_{j \in \beta}$ corresponds to MEV extraction abilities of non-buyers and $\Gamma: \mathbb{R}_+^{B + |\beta|} \mapsto \mathbb{R}_+$ maps all MEV extraction abilities to the price determined from the PBS mechanism. Note that the arbitrary specification of $\Gamma$ implies that we take no stand on how prices are formed through the PBS process; rather, we assume only that this price is available symmetrically across all ET buyers. Note also that, while the PBS price is symmetric across all ET buyers, the overall pay-off is not symmetric because each ET buyer has an outside option to build the block in which case the pay-off depends on the ET buyer's specific MEV extraction ability.

\subsubsection*{Benchmark Case: Naive PBS Implementation.}
In this case we establish two benchmark results assuming that builders do not buy ETs. These highlight the role of capital differentiation and are meant as a benchmark for the case where everyone buys ETs, analyzed in the subsequent subsection. \\

\begin{corollary}{ETs with PBS where Builders do not buy tickets}\\
For simplicity, we assume that buyers are risk-neutral (i.e., $\forall b: \Pi_b(x) = x$). Additionally, to examine the case that builders do not buy ETs, we assume that all buyers have no MEV extraction abilities $($i.e., $X_{b,t} = 0)$, implying that all buyers have the same pay-off, $R_{b,t} = \Gamma(\{0 \}_{i \in \mathcal{B}}, \{ Y_{j,t} \}_{j \in \beta})$. In this case, MEV capture by the protocol is given as follows:

$$
\chi = \frac{1}{1 + r_{(2)} \cdot N} < 1
$$
where $r_{(2)}$ is the second-lowest cost-of-capital among all buyers. Moreover, the equilibrium ET distribution is determined fully by cost-of-capital:

$$
k_b = \begin{cases}
    \frac{N}{|\overline{\overline{\mathcal{B}}}|}
& \text{ if } r_b = r_{(1)}\\
0 & \text{if } r_b > r_{(1)}\end{cases} 
$$
where $r_{(1)}$ is the lowest cost of capital among all buyers.
\label{prop:main6}
\end{corollary}

Corollary \ref{prop:main6} follows directly from Proposition \ref{prop:main4}. More explicitly, when ET buyers have no MEV extraction ability, then the pay-offs for all ET buyers are symmetric. In turn, the setting is a special case of that studied in Proposition \ref{prop:main4} and the corresponding results therefore follow. Crucially, under PBS, capital costs become particularly important. Indeed, even with no MEV extraction ability, a small set of investors could dominate the ET buyer market. To clarify that point, we offer the following further corollary:

\begin{corollary}
Assume that there exists a set of large investors, $\mathcal{I} \subset \mathcal{B}$, where this designation means that they possess no MEV extraction ability $($i.e., $X_{b,t} = 0$ for $b \in \mathcal{I})$ and possess a strictly lower cost of capital than all other investors (i.e., For all $b\in \mathcal{I}: r_b  < \min_{j \in \mathcal{B}/\mathcal{I}}~r_j$). Moreover, as before, assume that anyone with MEV extraction ability does not participate in buying tickets so that $X_{b,t} = 0$ for all $b \in \mathcal{B}$. 

Then, \textbf{all} ETs are purchased by large investors:

$$
\sum\limits_{b \in \mathcal{I}} k_b = N
$$
\label{prop:main7}
\end{corollary}

Corollary \ref{prop:main7} establishes that, with PBS, the ET market is dominated more so by institutions with financial advantages rather than building advantages. More explicitly, risk-neutral entities with the lowest capital costs would hold \textit{all} ETs even without \textit{any} MEV extraction abilities. The intuition is straight-forward: given an PBS market, any investor may purchase an ET and sell to the buyer market during the slot thereby nullifying the need for any MEV extraction skills. Nonetheless, buying ETs requires locking up capital and taking risk and thus investors with low cost of capital and high risk tolerance are likely to dominate ET holdings when there exists an PBS market.

\subsubsection*{Full Case: Builders buy ETs.}

We now consider the case where builders participate in PBS and show that MEV extraction ability and capital costs both play a crucial role in determining the distribution of ET holders.

\begin{proposition}{ET Buyers with PBS}\\
For simplicity, we assume that buyers are risk-neutral, $\forall b: \Pi_b(x) = x$. In that case, all tickets are bought by buyers with the highest valuation, $\overline{P}_{(1)}$, which is given explicitly as follows:

$$
\overline{P}_{(1)} = \underset{b \in \mathcal{B}}{\max}~\frac{\mathbb{E}[R_{b,t}]}{1+r_b \cdot N}
$$

More formally, we have the following result:

$$
\sum\limits_{b \in \overline{B}} k_b = N
$$

where $\overline{B}$, defined in Equation \eqref{eqref:bestbuyerset}, is the set of buyers with valuation $\overline{P}_{(1)}$.

\label{prop:main8}
\end{proposition}

Proposition \ref{prop:main8} demonstrates that, when builders participate in PBS, then the holdings of ETs is determined both by the builder's MEV extraction abilities and the capital costs of buyers more generally. More explicitly, each buyer's valuation is increasing in their pay-off which depends asymmetrically on their MEV extraction ability (see Equation \ref{eqn:investormevpayoff}). Nonetheless, each buyer's valuation also depends on her cost of capital. Ultimately, only those buyers with the highest valuation, $\overline{P}_{(1)}$, will purchase ETs.

To highlight the significance of capital costs, we conclude with the following result that highlights the possibility that large investors, and not builders, dominate the ET buyer market:

\begin{proposition}
Large Investors Dominate with PBS\\
Suppose that there exists a set of large investors, $\mathcal{I} \subset \mathcal{B}$, where this designation means that each member of the set possesses no MEV extraction ability $($i.e., $X_{b,t} = 0$ for $b \in \mathcal{I})$ but possesses especially low capital costs. In particular, we assume that their capital costs are sufficiently low that $$\forall i \in \mathcal{I}: r_i  < \frac{1}{N} \min_{b \in \mathcal{B}/\mathcal{I}}~ \Big( (1 + r_b \cdot N) \times \frac{\mathbb{E}[R_{\mathcal{I},t}]}{\mathbb{E}[R_{b,t}]} - 1\Big)$$ where $R_{\mathcal{I},t} = \Gamma\big(\{X_{i,t} \}_{i \in \mathcal{B}}, \{Y_{j,t} \}_{j \in \beta}\big)$ denotes the pay-off from buying an ET for each large investor. In this case, \textbf{all} ETs are purchased by large investors:

$$
\sum\limits_{b \in \mathcal{I}} k_b = N
$$

\label{prop:main9}
\end{proposition}

Proposition \ref{prop:main9} establishes that, given sufficiently low capital costs for investors with no MEV extraction ability, then these investors would dominate the ET buyer market. Buying ETs to either sell through PBS or to build during the eventual slot requires holding ETs in advance. In turn, the cost of holding ETs is a capital cost related to the cost of funding. When large investors are sufficiently advantaged with regard to funding costs, this advantage may exceed the advantage of builders and thus large investors could dominate the ET buyer market. 

\section{Discussion and Extensions}
\label{sec:discussion}

Our primary findings indicate that, without PBS, ETs can fully capture MEV when buyers are homogeneous, risk-neutral, and face no capital costs. However, the efficiency of MEV extraction diminishes with increased risk aversion and capital costs. Additionally, heterogeneity among buyers leads to lower MEV capture, with the potential for a single buyer to dominate the market. The presence of a PBS mechanisms can mitigate the centralization but not fully, MEV extraction ability gives buyers an advantage even with PBS. Moreover, PBS introduces another centralization vector whereby, investors with low capital cost but no extraction ability, may come to dominate the ET market.

\subsection{Interpreting the Results in Practice}

In practice, buyers are heterogeneous with varying capital costs, varying levels of risk aversion, and differences in abilities to extract MEV. This is evidenced by 85\% of blocks currently being built by just three actors \cite{mev_pics}. This concentration is not merely due to a higher randomized block win rate but also a result of differences in MEV extraction abilities among buyers (i.e., for these buyers $\mathbb{E}[R_{b,t}] \geq \mathbb{E}[R_{b^\prime,t}]$  for $b \neq b'$). In the current MEV-Boost (PBS) system, ex-post bids are solicited for the right to compose a block, with the highest bid winning. This means that the winning builder has the highest ability to extract MEV from a block, and about 94\% of blocks built using MEV-Boost are built by these three dominant actors \cite{mev_pics}, highlighting their superior MEV extraction capabilities. 

Therefore, it is reasonable to assume that there will be builder heterogeneity when ETs are implemented. Moreover, it is likely that even after ETs are implemented PBS will continue to be popular. Hence, the results from Proposition \ref{prop:main8} and  Proposition \ref{prop:main9} are most likely what will be seen in practice. Namely, ETs will be owned by those with a balance of the lowest capital costs and the best MEV extraction ability. There is also the potential for large investors to dominate the market, assuming they have far lower capital costs than builders. This does not preclude the possibility that Execution Payloads are ultimately constructed by those with the best ex-post MEV extraction ability.

A critical consideration is whether ETs offer a superior solution compared to the current system. As explained in the background section, the objective of these new MEV capture mechanisms is to enhance the allocation of MEV and mitigate overall centralization of the validator set.

\subsubsection{Allocation}
From an allocation perspective, ETs capture more MEV than the current system, which currently captures zero MEV. Even though ETs capture net more than the existing paradigm, there could be even better solutions that can capture more MEV, such as some version of Execution Auctions (EAs). More work is needed to analyze whether those solutions are superior.

\subsubsection{Centralization}
Regarding centralization, our research sheds light on the specific centralization vectors introduced by ETs as a function of MEV extraction and cost of capital differences. This remains true even in the presence of PBS as it is specified today. Looking at the roles of different participants separately, we can conclude that: (1) ETs allow the protocol to capture and distribute MEV which reduces centralization pressure on the validator set, but further research is required to determine if new timing games emerge as a result of this structural change; (2) ETs introduce ET holders a new actor to the MEV pipeline, we show that tickets are likely to be concentrated among a subset of buyers (more research is needed to assess the effects of that centralization in practice); (3) ETs do not alleviate centralization in the builder market.

\subsection{Limitations and Extensions}

\subsubsection{Exogenous Factors} The model does not account for time varying exogenous factors that can affect buyers. For instance, macroeconomic events such as interest rate cuts can alter the cost of capital for buyers. Additionally, buyers' abilities to extract MEV might vary in different environments. For example, some buyers might perform better when long-tail assets are more volatile, while others might excel with short-term asset volatility. Moreover, builders receive proprietary order flow, and changes in the quality of this flow can impact their MEV extraction ability. Future implementations of the model should incorporate a state variable representing the state of the world at time \( t \).

\subsubsection{Relaxing the i.i.d. Assumption}

For simplicity, we assumed that MEV extraction for each builder is i.i.d. across time and thus the endogenous stationary price of ETs is time-invariant. However, in practice, MEV is not i.i.d. and is often correlated over short increments, as MEV is related to market volatility \cite{LVR}. When MEV is high (low) in one block, adjacent blocks are likely to have high (low) MEV as well. By incorporating a state variable into the model, we could relax the i.i.d. assumption. This state variable would capture the temporal correlation of MEV and reflect the influence of market conditions on MEV extraction. 

\subsubsection{Time-invariance of the Price}

The assumption that the price of an ET is time-invariant fails if MEV is not i.i.d. or if there are time-varying exogenous influences. A future model with a relaxed i.i.d. assumption and/or includes a state variable will need to account for a temporally correlated and exogenously influenced price.

\subsubsection{Multi-block MEV}
The current model does not consider Multi-block MEV (MMEV).\footnote{Controlling consecutive blocks may yield higher MEV than the sum of the MEV from controlling individual blocks.} This introduces an endogenous centralization vector where a builder who controls block \( t \) might bid more aggressively for block \( t+1 \) as consecutive control could yield more MEV. This endogenous factor could significantly affect both allocation and centralization outcomes in the model. More research is needed to study MMEV and its effects on ETs.

\subsubsection{Mechanism for Selling ETs}
The model does not define the market mechanisms for selling Execution Tickets. The method of selling these tickets can influence builder behavior and potentially introduce new timing games and externalities. Future research should develop a comprehensive model for the sale of Execution Tickets and study its impact on the overall system dynamics.


\section*{Acknowledgments}
The authors acknowledge helpful discussions and comments from Brad Bachu, Joel Hasbrouck, Ruizhe Jia, Julian Ma and Xin Wan.

\appendix

\section*{Appendices}
\addcontentsline{toc}{section}{Appendices}

\section{Equilibrium Definition and Supplementary Lemmas}
\label{app:eqdefandlemmas}
\renewcommand{\theequation}{\thesection.\arabic{equation}}
\setcounter{equation}{0}

As is standard, our equilibrium definition requires that builders behave optimally and that the supply of ETs equates with its demand. More formally, the definition of equilibrium is given as follows:
\begin{definition} Equilibrium Definition\\
An equilibrium is a set of builder ET holdings, $\{ k_b \}_{b: b \in \mathcal{B}}$ and a price for ETs, $P$, such that the following conditions hold:
\begin{enumerate}
    \item \underline{All builder ET holdings are optimal}\par
    $\forall b \in \mathcal{B}: k_b$ solves Equation \eqref{eqn:ETpurchase}
    \smallskip
    \item \underline{ET price is such that ET market demand equates with supply}\par
    $\sum\limits_{b \in \mathcal{B}} k_b = N$
\end{enumerate}
\label{def:equilibrium}
\end{definition}

\begin{lemma} Optimal Builder ET Holdings\\
For any Builder $b \in \mathcal{B}$, the optimal ET holding, $k_b$, is given as follows when $P \neq \overline{P}_b$:
\begin{equation}
    k_b = \begin{cases} 0 & \text{ if } P > \overline{P}_b \\
        N & \text{ if } P < \overline{P}_b \\
    \end{cases}
\end{equation}
Moreover, when $P = \overline{P}_b$, then any feasible ET holding (i.e., any $k_b \in \{0,...,N\}$) is optimal.
\label{lemma:optkb}
\end{lemma}

\begin{proof}~\\
By definition, $k_b$ is the solution to Equation \eqref{eqn:ETpurchase}. More formally:
\begin{equation}
k_b \in \underset{k_b \in \{0,...,N\} }{\arg\max}~ \mathbb{E}[\Pi_b(\text{P\&L}_{b,t})] - r_{b} \cdot P \cdot k_b
\end{equation}
Then, expanding the expectation in the objective function and applying Equation \eqref{eqn:payoff} yields the following:
\begin{equation}
k_b \in \underset{k_b \in \{0,...,N\} }{\arg \max}~ \frac{k_b}{N}\Pi_b(R_{b,t} - P) - r_{b} \cdot P \cdot k_b
\end{equation}
Crucially, the objective function is linear in $k_b$ with coefficient $\frac{1}{N}\Pi_b(R_{b,t} - P) - r_{b} \cdot P$. Moreover, by definition of $\overline{P}_b$ (see Equation \ref{eqn:maxPb}), the coefficient is zero at $P = \overline{P}_b$, thereby implying the last part of the result, that any feasible ET holding is optimal whenever $P = \overline{P}_b$. Moreover, $\Pi_b^\prime > 0$ and $r_b \geq 0$ imply that the coefficient is strictly decreasing $P$ so that it is strictly negative whenever $P > \overline{P}_b$ and strictly positive whenever $P < \overline{P}_b$. In turn, the strictly negative coefficient when $P > \overline{P}_b$ implies $k_b = 0$ whenever $P > \overline{P}_b$ and the strictly positive coefficient when $P < \overline{P}_b$ implies $k_b = N$ whenever $P < \overline{P}_b$, thereby completing the proof.
\end{proof}

\begin{lemma} Necessary Condition for Equilibrium\\
$P \in [\overline{P}_{(2)}, \overline{P}_{(1)}]$ is a necessary condition for equilibrium.
\label{lemma:necconditions}
\end{lemma}
\begin{proof}~\\
If $P < \overline{P}_{(2)}$, then Lemma \ref{lemma:optkb} implies $k_b = N$ for at least two builders and thus the second condition of Equilibrium Definition \ref{def:equilibrium} cannot hold, implying equilibrium cannot arise. Moreover, if $P > \overline{P}_{(1)}$, then Lemma \ref{lemma:optkb} implies $k_b = 0$ for all builders and thus the second condition of Equilibrium Definition \ref{def:equilibrium} cannot hold, implying equilibrium cannot arise. Putting the two aforementioned statements together, $P \notin [\overline{P}_{(2)}, \overline{P}_{(1)}]$ implies equilibrium cannot arise. Taking the contrapositive then implies the desired result, namely that an equilibrium implies $P \in [\overline{P}_{(2)}, \overline{P}_{(1)}]$.
    
\end{proof}

\begin{lemma} Sufficient Condition for Equilibrium\\
$P \in [\overline{P}_{(2)}, \overline{P}_{(1)}]$ is a sufficient condition for equilibrium in the sense that there always exist ET holdings, $\{ k_b \}_{b \in \mathcal{B}}$, consistent with $P \in [\overline{P}_{(2)}, \overline{P}_{(1)}]$ such that $P$ and $\{ k_b \}_{b \in \mathcal{B}}$ satisfy the Equilibrium Definition \eqref{def:equilibrium}.
\label{lemma:suffconditions}
\end{lemma}
\begin{proof}~\\
We provide a constructive proof. More explicitly, direct verification, using Lemma \ref{lemma:optkb}, reveals that the following solution always satisfies Equilibrium Definition \eqref{def:equilibrium}:
\begin{equation}
    P = \overline{P}_{(2)}, \qquad \forall b \in \mathcal{B}: k_b = \frac{N \cdot \mathcal{I}(b \in \overline{\mathcal{B}})}{|\overline{\mathcal{B}}|}
\end{equation}
where $\overline{\mathcal{B}}$ is defined in Equation \eqref{eqref:bestbuyerset} and $|X|$ refers to the cardinality of the set X. As an aside, we emphasize that this constructed solution implies that one builder holds all ETs whenever $\overline{P}_{(2)} < \overline{P}_{(1)}$. In particular, in that case, $|\overline{\mathcal{B}}| = 1$. 
\end{proof}
  
\begin{lemma} Necessary Condition for Equilibrium II\\
The following condition must hold in any equilibrium: $\sum\limits_{b:b \in \overline{\mathcal{B}}} k_b = N$
\label{lemma:necconditions2}
\end{lemma}
\begin{proof}~\\
Lemma \eqref{lemma:necconditions} implies that $P \geq P_{(2)}$ in any equilibrium and thus there are only three possible cases: (i) $P > \overline{P}_{(2)}$, (ii) $P = \overline{P}_{(2)} = \overline{P}_{(1)}$ and (iii) $P = \overline{P}_{(2)} < \overline{P}_{(1)}$. We prove the result in each case separately below.\\

Case (i): $b \notin \overline{\mathcal{B}} \implies \overline{P}_b \leq \overline{P}_{(2)} < P \implies k_b = 0$ where the last implication follows from Lemma \ref{lemma:optkb}. Then, via the second part of Definition \eqref{def:equilibrium}, $\sum_{b:b \in \overline{\mathcal{B}}} k_b = N - \sum_{b:b \notin \overline{\mathcal{B}}} k_b = N - 0 = N$ as desired. \\

Case (ii): $b \notin \overline{\mathcal{B}} \implies \overline{P}_b < \overline{P}_{(1)} = \overline{P}_{(2)} = P \implies k_b = 0$ where the last implication follows from Lemma \ref{lemma:optkb}. Then, via the second part of Definition \eqref{def:equilibrium}, $\sum_{b:b \in \overline{\mathcal{B}}} k_b = N - \sum_{b:b \notin \overline{\mathcal{B}}} k_b = N - 0 = N$ as desired. \\

Case (iii): Lemma \ref{lemma:optkb} yields that $b \in \overline{\mathcal{B}} \implies \overline{P}_b = \overline{P}_{(1)} > \overline{P}_{(2)} = P \implies  k_b = N$. Then, via the second part of Definition \eqref{def:equilibrium}, $\sum_{b:b \notin \overline{\mathcal{B}}} k_b = N - \sum_{b:b \in \overline{\mathcal{B}}} k_b \leq N - N = 0$. Then, since $k_b \geq 0$ for all $b \in \mathcal{B}$, the previous result $\sum_{b:b \notin \overline{\mathcal{B}}} k_b \leq 0$ therefore implies $\sum\limits_{b:b \notin \overline{\mathcal{B}}} k_b = 0$. Finally, the second part of Definition \eqref{def:equilibrium} yields $\sum_{b:b \in \overline{\mathcal{B}}} k_b = N - \sum_{b:b \notin \overline{\mathcal{B}}} k_b = N - 0 = N$ as desired.
\end{proof}

\section{Proofs of Results from Section \ref{sec:results}}
\label{app:proofs}
\renewcommand{\thesubsection}{\thesection.\arabic{subsection}}
\setcounter{subsection}{0}

\renewcommand{\theequation}{\thesection.\arabic{equation}}
\setcounter{equation}{0}

\subsection{Proof of Proposition \ref{prop:eqsoln}}
This result follows directly from Lemmas \ref{lemma:necconditions} and \ref{lemma:suffconditions}.

\subsection{Proof of Proposition \ref{prop:main1}}
Applying $\Pi_b(x) = x$, $r_b = 0$ and $R_{b,t} = R_{\overline{B}}$ for all $b \in \mathcal{B}$ to Equation \eqref{eqn:maxPb} yields:

\begin{equation}
\forall b \in \mathcal{B}: \overline{P}_b = \max\{ P : \frac{1}{N} \Big( \mathbb{E}[R_{\overline{B}}] - P \Big) \geq 0\}
\label{eqn:maxPbriskneutral}
\end{equation}
which further implies:
\begin{equation}
    \forall b \in \mathcal{B}: \overline{P}_b = \mathbb{E}[R_{\overline{B}}] 
\end{equation}
Applying this result to Proposition \ref{prop:eqsoln} implies $P = \mathbb{E}[R_{\overline{B}}]$ which implies the first part of the result, $\chi = 1$. Finally, as per Lemma \eqref{lemma:optkb}, any $k_b \in \{1,...,N\}$ is optimal for Builder $b$ and thus, given the second requirement in Equilibrium Definition \eqref{def:equilibrium}, we impose the symmetric solution $k_b = \frac{N}{B}$ for all $b \in \mathcal{B}$. 

\subsection{Proof of Proposition \ref{prop:main2}}
Jensen's inequality implies:
\begin{equation}
    \mathbb{E}[\Pi_b(R_{\overline{B}} - P)] < \Pi_b(\mathbb{E}[R_{\overline{B}}] - P)
\end{equation}
Thus, if $P \geq \mathbb{E}[R_{\overline{B}}]$, then $\Pi^\prime > 0$ and $\Pi(0) = 0$ imply $\mathbb{E}[\Pi_b(R_{\overline{B}} - P)] < \Pi_b(\mathbb{E}[R_{\overline{B}}] - P) < 0$ and further $\frac{1}{N}\mathbb{E}[\Pi_b(R_{\overline{B}} - P)] - r_b\cdot P < 0$. In turn, $\max\{ P : \frac{1}{N} \mathbb{E}[\Pi_b(R_{b,t}- P)] - r_b \cdot P \geq 0\} \subseteq [0, \mathbb{E}[R_{\overline{B}}])$ and thus $\{ P : \frac{1}{N} \mathbb{E}[\Pi_b(R_{b,t}- P)] - r_b \cdot P \geq 0\}$ being closed implies $\overline{P}_{(b)} < \mathbb{E}[R_{\overline{B}}]$ which implies the first part of the result, $\chi < 1$.\footnote{$\{ P : \frac{1}{N} \mathbb{E}[\Pi_b(R_{b,t} - P)] - r_b \cdot P \geq 0\}$ being a closed set follows from $\Pi_b^{\prime\prime} < 0$ and $\mathbb{E}[R_{b,t}] < \infty$.}

The second part of the result follows from symmetry implying $\overline{P}_b$ being equal for all Builders $b \in \mathcal{B}$. In turn, Equilibrium Definition \eqref{def:equilibrium} implies $P = \overline{P}_b$ for all Builders $b \in \mathcal{B}$. Then, as per Lemma \eqref{lemma:optkb}, any $k_b \in \{1,...,N\}$ is optimal for Builder $b$ and thus, given the second requirement in Equilibrium Definition \eqref{def:equilibrium}, we impose the symmetric solution $k_b = \frac{N}{B}$ for all $b \in \mathcal{B}$.  

\subsection{Proof of Proposition \ref{prop:main3}}

Applying $\Pi_b(x) = x$ and $r_b = 0$ to Equation \eqref{eqn:maxPb} yields:

\begin{equation}
\overline{P}_b = \max\{ P : \frac{1}{N} \Big( \mathbb{E}[R_{b,t}] - P \Big) \geq 0\}
\label{eqn:maxPbheteroegeneous}
\end{equation}
which further implies:
\begin{equation}
    \overline{P}_b = \mathbb{E}[R_{b,t}]
\end{equation}
Moreover:
\begin{equation}
    \overline{P}_{(1)} = \mathbb{E}[R_{1,t}],\qquad \overline{P}_{(2)} = \mathbb{E}[R_{2,t}]
\end{equation}
In turn, direct verification from Definition \eqref{def:equilibrium} reveals that the following is an equilibrium:
\begin{equation}
    P = \overline{P}_{(2)} = \mathbb{E}[R_{2,t}],\qquad k_b = N \cdot \mathcal{I}(b = 1)
\end{equation}
which establishes the second part of the result. For the first part of the result, note that $k_b = N \cdot \mathcal{I}(b = 1)$ implies $R_{\overline{B}} = R_{1,t}$. In turn, the first part of the result holds as follows:
\begin{equation}
    \chi = \frac{P}{\mathbb{E}[R_{\overline{B}}]} = \frac{\mathbb{E}[R_{2,t}]}{\mathbb{E}[R_{1,t}]} < 1
\end{equation}

\subsection{Proof of Proposition \ref{prop:main4}}
By direct verification from Definition \eqref{def:equilibrium}, the following is an equilibrium:

\begin{equation}
    P = \frac{\mathbb{E}[R_{1,t}]}{1 + r_{(2)}\cdot N},\qquad k_b = \begin{cases}
    \frac{N}{|\overline{\overline{\mathcal{B}}}|}
& \text{ if } r_b = r_{(1)}\\
0 & \text{if } r_b > r_{(1)}\end{cases} 
\end{equation}
which establishes the last part of the result. For the first part, note that $\forall b, b^\prime R_{b,t} = R_{b^\prime,t}$ which implies that $forall b: \mathbb{E}[R_{b,t}] = \mathbb{E}[R_{\overline{B}}]$. Then, applying $P = \frac{\mathbb{E}[R_{\overline{B}}]}{1 + r_{(2)}\cdot N}$ to Equation \eqref{eqn:mevinternalizationratio} yields the desired result. 

\subsection{Proof of Proposition \ref{prop:main8}}

By direct verification from Definition \eqref{def:equilibrium}, the following is an equilibrium:

\begin{equation}
    P = \overline{P}_{(2)},\qquad k_b = \begin{cases}
    \frac{N}{|\overline{\mathcal{B}}|}
& \text{ if } b \in \overline{\mathcal{B}}\\
0 & \text{otherwise}\end{cases} 
\end{equation}

and thus an equilibrium exists. The result $\sum\limits_{b \in \overline{B}} k_b = N$ then follows as a corollary of Lemma \ref{lemma:necconditions2}.

\subsection{Proof of Proposition \ref{prop:main9}}

By direct verification from Definition \eqref{def:equilibrium}, the following is an equilibrium:

\begin{equation}
    P = \overline{P}_{(2)},\qquad k_b = \begin{cases}
    \frac{N}{|\overline{\mathcal{B}}|}
& \text{ if } b \in \overline{\mathcal{B}}\\
0 & \text{otherwise}\end{cases} 
\end{equation}

and thus an equilibrium exists. Then, $\forall i \in \mathcal{I}: r_i  < \frac{1}{N} \min_{b \in \mathcal{B}/\mathcal{I}}~ \Big( (1 + r_b \cdot N) \times \frac{\mathbb{E}[R_{\mathcal{I},t}]}{\mathbb{E}[R_{b,t}]} - 1\Big)$ implies $\mathcal{B} \subseteq \mathcal{I}$ so that the result $\sum\limits_{b \in \overline{I}} k_b = N$ follows as a corollary of Lemma \ref{lemma:necconditions2}.

\bibliographystyle{splncs04}
\bibliography{references}

\appendix

\end{document}